\documentclass[amsthm,titlepage]{elsart}
\pdfoutput=1

\makeatletter
\let\AND\@undefined
\def\ps@copyright{\let\@mkboth\@gobbletwo
  \def\@oddhead{}%
  \let\@evenhead\@oddhead
  \def\@oddfoot{\small\slshape
   To appear: Journal of Symbolic Computation (JSC), 2011. \hfil
    }%
  \let\@evenfoot\@oddfoot
}

\makeatother

\usepackage{yjsco}
\usepackage[round]{JSCnatbib}
\usepackage{url}

\usepackage{graphicx}
\usepackage{color}
\usepackage[pdfauthor={Mark Giesbrecht and Daniel S. Roche},
	pdftitle={On Lacunary Polynomial Perfect Powers},
	colorlinks,linkcolor=darkblue,citecolor=darkgreen,urlcolor=darkred]{hyperref}

\definecolor{darkgreen}{rgb}{0,.35,0}
\definecolor{darkblue}{rgb}{0,0,.5}
\definecolor{darkred}{rgb}{.6,0,0}

\usepackage{latexsym}              
\usepackage{amssymb}               
\usepackage{amsmath}               
\usepackage{amsfonts}

\newcommand{\F}{{\mathsf{F}}}
\newcommand{\ftil}{{\tilde{f}}}
\newcommand{\fbar}{{\overline{f}}}

\newcommand{\norm}[1]{{\|#1\|}}
\newcommand{\inorm}[1]{\|#1\|_\infty}
\newcommand{\calE}{{\mathcal{E}}}
\newcommand{\calP}{{\mathcal{P}}}
\newcommand{\calS}{{\mathcal{S}}}
\newcommand{\calC}{{\mathcal{C}}}

\renewcommand{\H}{{\mathcal H}}
\newcommand{\ZZ}{{\mathbb Z}}
\newcommand{\NN}{{\mathbb N}}
\newcommand{\QQ}{{\mathbb Q}}

\newcommand{\RR}{{\mathbb R}}
\newcommand{\CC}{{\mathbb C}}
\newcommand{\nonzero}{{\setminus\{0\}}}
\newcommand{\gf}{{\mathbb F}}
\newcommand{\nin}{{\mathrel{\hbox{$\mskip4mu\not\in\mskip4mu$}}}}
\newcommand{\disc}{{\mathop{\rm disc}}}
\newcommand{\lcoeff}{{\mathop{\rm lcoeff}}}
\newcommand{\res}{{\mathop{\rm res}}}
\newcommand{\ndivs}{{\mskip3mu\nmid\mskip3mu}}
\newcommand{\nequiv}{{\mathrel{\hbox{$\mskip4mu\not\equiv\mskip4mu$}}}}
\newcommand{\divs}{{\mskip3mu|\mskip3mu}}
\newcommand{\softO}{{O\,\tilde{}\,}}
\renewcommand\rem{\operatorname{rem}}
\newcommand{\xbar}{{\overline x}}
\newcommand{\ebar}{{\overline e}}
\newcommand{\tms}[1]{\tau(#1)}
\newcommand{\chizero}{{\chi_{\lower2pt\hbox{$\scriptscriptstyle 0$}}}}
\newcommand{\ceil}[1]{\lceil#1\rceil}

\newcommand{\M}{{\mbox{M}}}
\DeclareMathOperator{\lc}{lc}
\DeclareMathOperator{\characteristic}{char}

\renewcommand{\thefootnote}{\fnsymbol{footnote}}

\newtheorem{theorem}{Theorem}[section]
\newtheorem{lemma}[theorem]{Lemma}
\newtheorem{corollary}[theorem]{Corollary}

\newtheorem{conjecture}[theorem]{Conjecture}

\usepackage{comment}
\usepackage{float}
\usepackage[noend]{algorithmic}
\usepackage{ifthen}
\floatstyle{ruled} \newfloat{algg}{H}{loa}
\floatname{algg}{Algorithm}
\renewcommand{\algorithmicrequire}{\textbf{Input:}}
\renewcommand{\algorithmicensure}{\textbf{Output:}}

\renewenvironment{alg}[2][nolabel]%
{\renewcommand{\thealgg}{{\tt {#2}}}
\begin{algg}
\ifthenelse{\equal{#1}{nolabel}}{\caption{ }}{\caption{ \label{#1}}}
\begin{algorithmic}[1]}%
{\end{algorithmic}\end{algg}}

\numberwithin{equation}{section}
\renewcommand{\theequation}{\arabic{section}.\arabic{equation}}

\numberwithin{table}{section}

\DeclareMathAlphabet{\mathitbf}{OML}{cmm}{b}{it}

\makeatletter
\renewenvironment{proof}%
  {\par\addvspace{\@bls \@plus 0.5\@bls \@minus 0.1\@bls}\noindent
   {\bfseries\Elproofname}\enspace\ignorespaces}%
  {\qed\par\addvspace{\@bls \@plus 0.5\@bls \@minus 0.1\@bls}}
\newenvironment{proof*}%
  {\par\addvspace{\@bls \@plus 0.5\@bls \@minus 0.1\@bls}\noindent
   {\bfseries\Elproofname}\enspace\ignorespaces}%
  {\par\addvspace{\@bls \@plus 0.5\@bls \@minus 0.1\@bls}}
\makeatother

\journal{Journal of Symbolic Computation (JSC)}
\begin{document}

\begin{frontmatter}

\title{Detecting lacunary perfect powers and computing their roots}

\author{Mark Giesbrecht}
\address{Cheriton School of Computer Science,
University of Waterloo,
Waterloo, Ontario, Canada}
\ead{mwg@cs.uwaterloo.ca}

\author{Daniel S. Roche}
\address{Cheriton School of Computer Science,
University of Waterloo, Waterloo, Ontario, Canada}
\ead{droche@cs.uwaterloo.ca}

\begin{abstract}
  We consider solutions to the equation $f = h^r$ for polynomials $f$
  and $h$ and integer $r\geq 2$. Given a polynomial $f$ in the lacunary
  (also called \emph{sparse} or \emph{super-sparse}) representation, we
  first show how to determine if $f$ can be written as $h^r$ and, if so, 
  to find such an $r$. This is a Monte Carlo randomized algorithm whose
  cost is polynomial in the number of non-zero terms of $f$ and in 
  $\log \deg f$, i.e., polynomial in the size of the lacunary
  representation, and it works over $\gf_q[x]$ (for large
  characteristic) as well as $\QQ[x]$. We also give two deterministic
  algorithms to compute the perfect root $h$ given $f$ and $r$. The
  first is output-sensitive (based on the sparsity of $h$) and works
  only over $\QQ[x]$. A sparsity-sensitive Newton iteration forms the
  basis for the second approach to computing $h$, which
  is extremely efficient and works over both $\gf_q[x]$ (for large
  characteristic) and $\QQ[x]$, but depends
  on a number-theoretic conjecture. Work of Erd\"os, Schinzel, Zannier,
  and others suggests that both of these algorithms are unconditionally
  polynomial-time in the lacunary size of the input polynomial $f$.
  Finally, we demonstrate the efficiency of the randomized detection
  algorithm and the latter perfect root computation algorithm with an
  implementation in the \verb|C++| library NTL.
\end{abstract}

\end{frontmatter}

\section{Introduction}

In this paper we consider the problem of determining whether a
polynomial $f$ equals $h^r$ for some other polynomial $h$ and 
integer $r\geq 2$, and if so, finding $h$ and $r$.  The novel aspect of this
current work is that our algorithms are efficient for the
\emph{lacunary} (also called \emph{sparse} or \emph{supersparse})
representation of polynomials.  Specifically, we write
\begin{equation}
  \label{eq:lacpolymv}
  f=\sum_{1\leq i\leq t} c_i \xbar^{\,\ebar_i} \in \F[x_1,\ldots,x_\ell],
\end{equation}
where $\F$ is a field, $c_0,\ldots,c_t\in\F\nonzero$, 
$\ebar_1,\ldots,\ebar_t\in \NN^\ell$ are distinct exponent tuples with $0\leq
\norm{\ebar_1}_1\leq \cdots\leq \norm{\ebar_t}_1=\deg f$, and
$\xbar^{\,\ebar_i}$ is the monomial $x_1^{e_{i1}}x_2^{e_{i2}}\cdots
x_\ell^{e_{i\ell}}$ of degree $\norm{\ebar_i}_1=\sum_{1\leq
  j\leq\ell}e_{ij}$.  We say $f$ is \emph{$t$-sparse} and write
$\tms{f}=t$.  We present algorithms which require time polynomial in
$\tms{f}$ and $\log\deg f$.

Computational work on \emph{lacunary} polynomials has proceeded
steadily for the past three decades.  From the dramatic initial
intractability results of \cite{Pla77,Pla84}, through progress in
algorithms (e.g., \cite{BenTiw88,Shp00,KalLee02}) and complexity
(e.g., \cite{KarShp99,Qui86,GatKar93}), to recent breakthroughs in
root finding and factorization \citep{CucKoi99,KalKoi06,Len97}, these
works have important theoretical and practical consequences.
The lacunary representation is arguably more intuitive than the standard
dense representation, and in fact corresponds to the default linked-list
representation of polynomials in modern computer algebra systems such as
Maple and Mathematica.

We will always assume that $\tms{f}\geq 2$; otherwise $f=x^n$, and
determining whether $f$ is a perfect power is equivalent to
determining whether $n\in\NN$ is composite, and to factoring $n$ if we
wish to produce $r$ dividing $n$ such that $f=(x^{n/r})^r$.
Surprisingly, the intractability of the latter problem is avoided when
$\tms{f}\geq 2$.

We first consider detecting perfect powers and computing the power $r$
for the univariate case
\begin{equation}
  \label{eq:lacpolyuni}
  f=\sum_{1\leq i\leq t} c_i x^{e_i} \in \F[x],
\end{equation}
where $0\leq e_1<e_2<\cdots<e_t=\deg f$.  

Two cases for the field $\F$ are handled: the integers and finite
fields of characteristic $p$ greater than the degree of $f$.  When
$f\in\ZZ[x]$, our algorithms also require time polynomial in
$\log\norm{f}_\infty$, where $\norm{f}_\infty=\max_{1\leq i\leq t}
|c_i|$ (for $f\in\QQ[x]$, we simply work with $\fbar=c f\in\ZZ[x]$,
for the smallest $c\in\ZZ\nonzero$).  This reflects the bit-length of
coefficients encountered in the computations.  
Efficient techniques will also be
presented for reducing the multivariate case to the univariate one, and
for computing a root $h$ such that $f=h^r$.

\subsection{Related work and methods}

Two well-known techniques can be applied to the problem of testing for
perfect powers, and both are very efficient when $f=h^r$ is dense.  We
can compute the squarefree decomposition of $f$ as in \citep{Yun76},
and determine whether $f$ is a perfect power by checking whether the
greatest (integer) common divisor 
of the exponents of all nontrivial factors in the squarefree
decomposition is at least 2.  An even faster method (in theory and
practice) to find $h$ given $f=h^r$ is by a Newton iteration. This
technique has also proven to be efficient in computing perfect roots
of (dense) multi-precision integers \citep{BacSor93,Ber98}.  In
summary however, we note that both these methods require approximately
linear time in the \emph{degree} of $f$, which may be exponential in
the lacunary size.

Newton iteration has also been applied to finding perfect polynomial
roots of lacunary (or other) polynomials given by straight-line
programs. \cite{Kal87} shows how to compute a straight-line program
for $h$, given a straight-line program for $f=h^r$ and the value of
$r$.  This method has complexity polynomial in the size of the
straight-line program for $f$, and in the degree of $h$, and in
particular is effective for large $r$.  We do not address the powerful
generality of straight-line programs, but do avoid the dependence on
the degree of $h$.

Closest to this current work, \cite{Shp00} shows how to recognize
whether $f=h^2$ for a lacunary polynomial $f \in \gf_q[x]$.  
Shparlinski uses
random evaluations and tests for quadratic residues.  How to determine
whether a lacunary polynomial is \emph{any} perfect power is posed as
an open question.

\subsection{Our contributions}

Given a lacunary polynomial $f\in\ZZ[x]$ with $\tms{f}\geq 2$ and
degree $n$, we first present an algorithm to compute an integer $r>1$
such that $f=h^r$ for some $h \in \ZZ[x]$, or determine that
no such $r$ exists.  The algorithm requires
$\softO(t\log^2\inorm{f}\log^2n)$ machine operations%
\footnote{We employ soft-Oh notation: for functions $\sigma$ and
  $\varphi$ we say $\sigma\in \softO(\varphi)$ if $\sigma\in O(\varphi
  \log^c\varphi)$ for some constant $c\geq 0$.},
and is probabilistic of the Monte Carlo type.  That is, 
for any input, on any execution the probability of producing an
incorrect answer is strictly less than $1/2$, assuming the ability to
generate random bits at unit cost.
This possibility of error can be
made arbitrarily small with repeated executions.
Moreover, the error is one-sided, so we prove specifically
that deciding whether a given multivariate
rational polynomial encoded in the lacunary representation is a
perfect power is in the complexity class \textsf{coRP}.

A similar algorithm is presented to answer Shparlinski's open question 
on perfect powers of lacunary polynomials over finite
fields, at least for the case of large characteristic. That is, when
the characteristic $p$ of a finite field $\F$ is greater than $\deg
f$, we provide a Monte Carlo algorithm that determines if there exists
an $h\in\F[x]$ and $r$ such that $f=h^r$, and finds $r$ if it exists,
which requires $\softO(t\log^2n)$ operations in $\F$.

An implementation of our algorithm over $\ZZ$ in NTL indicates
excellent performance on sparse inputs when compared to a fast
implementation based on previous technology (a variable-precision
Newton iteration to find a power-series $r$th root of $f$, followed by
a Monte Carlo correctness check).

Actually computing $h$ such that $f=h^r$ is a somewhat trickier
problem, at least insofar as bounds on the sparsity of $h$ have not
been completely resolved.  Conjectures of \cite{Sch87} and recent work
of \cite{Zan07Acta} suggest that, provided the characteristic of $\F$
is zero or sufficiently large, $h$ is lacunary as well.
To avoid this lack of sufficient theoretical understanding, we develop
an algorithm which requires time polynomial in both the representation
size of the input $f$ (i.e., $\tau(f)$, $\log n$ and $\log\inorm{f}$)
\emph{and} the representation size of the output (i.e., $\tau(h)$ and
$\log\inorm{h}$).  This algorithm works by projecting $f$ into a
sequence of small cyclotomic fields.  Images of the desired $h$ in
these fields are discovered by factorization over an algebraic
extension.  Finally, a form of interpolation of the sparse exponents
is used to recover the global $h$.  
Thanks to an efficient perfect-root certification, this algorithm is
deterministic and polynomial-time,
however we do
not claim it will be efficient in practice.  Instead, we also present
and analyze a simpler alternative based on a kind of sparse Newton iteration.
Subject to what we believe is a reasonable conjecture, this is shown
to be very fast.

It may be helpful to point out the differences between the detection
algorithm in Section~\ref{sec:testpp} and the computation algorithms in
Section~\ref{sec:computepp}.
While the former is
probabilistic of the Monte Carlo type and does not actually produce
the perfect $r$th root $h$ of $f$ if it exists, it provably works in
polynomial-time even if $h$ is dense. 
The computation algorithms, by contrast, are deterministic but rely on
the unknown root $h$ being sparse.

The remainder of the paper is arranged as follows.  In Section 2 we
present the main theoretical tool for our algorithm to determine if
$f=h^r$, and to find $r$.  We also show how to reduce the multivariate
problem to the univariate one.  In Section 3 we show how to compute
$h$ such that $f=h^r$ (given that such $h$ and $r$ exist).  Finally,
in Section 4, we present an experimental implementation of some of our
algorithms in the \verb|C++| library NTL.

An earlier version of some of this work was presented in the ISSAC
2008 conference \citep{GieRoc08}.

\section{Testing for perfect powers}
\label{sec:testpp}

In this section we describe a method to determine if a lacunary
polynomial $f\in\F[x]$ is a perfect power.  That is, do there exist
$h\in\F[x]$ and $r>1$ such that $f=h^r$?  The polynomial $h$ need not
be lacunary, though some conjectures suggest it may well have to be.
We will find $r$, but not $h$.

We first describe algorithms to test if an $f\in\F[x]$ is an $r$th
power of some polynomial $h\in\F[x]$, where $f$ and $r$ \emph{are
  both given} and $r$ is assumed to be prime.  We present and analyze
variants that work over finite fields $\gf_q$ and over $\ZZ$.  In
fact, these algorithms for given $r$ are for \emph{black-box}
polynomials: they only need to evaluate $f$ at a small number of
points.  That this evaluation can be done quickly is a property of
lacunary and other classes of polynomials.

For lacunary $f$ we then show that, in fact, if $h$ exists at all then
$r$ must be small unless $f=x^n$. And if $f$ is a perfect power, then
there certainly exists a prime $r$ such that $f$ is an $r$th power. So
in fact the restrictions that $r$ is small and prime are sufficient to
cover all nontrivial cases, and our method is complete.

\subsection{Detecting given $r$th powers}

Our main tool in this work is the following theorem which says that,
with reasonable probability, a polynomial is an $r$th power if and
only if the modular image of an evaluation in a specially constructed
finite field is an $r$th power. 

\begin{theorem}
  \label{thm:fapow}
  Let $\varrho \in \ZZ$ be a prime power and $r\in\NN$ a prime
  dividing $\varrho-1$.  Suppose that $f\in\gf_\varrho[x]$ has degree
  $n\leq 1+\sqrt{\varrho}/2$ and is \emph{not} a perfect $r$th power
  in $\gf_\varrho[x]$.  Then
  \[
  R_{f}^{(r)}=\#\left\{c\in\gf_\varrho: f(c)\in\gf_\varrho~\mbox{is an $r$th
      power}\right\} \leq \frac{3\varrho}{4}.
  \]
\end{theorem}
\begin{proof}
  The $r$th powers in $\gf_\varrho$ form a subgroup $H$ of $\gf_\varrho^*$ of
  index $r$ and size $(\varrho-1)/r$ in $\gf_\varrho^*$.  Also, $a\in\gf_\varrho^*$ is
  an $r$th power if and only if $a^{(\varrho-1)/r}=1$.  We use the method of
  ``completing the sum'' from the theory of character sums.  We refer
  to \cite{LidNie83}, Chapter 5, for an excellent discussion of
  character sums.  By a multiplicative character we mean a
  homomorphism $\chi:\gf_\varrho^*\to\CC$ which necessarily maps $\gf_\varrho$
  onto the unit circle.  As usual we extend our multiplicative
  characters $\chi$ so that $\chi(0)=0$,
  and define the trivial character $\chizero(a)$ to be 0 when $a=0$ and
  1 otherwise.

  For any $a\in\gf_\varrho^*$,
  \[
  \frac{1}{r} \sum_{\chi^r=\chizero} \chi(a) =
  \begin{cases}
    1 & \mbox{if $a\in H$},\\
    0 & \mbox{if $a\nin H$},
  \end{cases}
  \]
  where $\chi$ ranges over all the multiplicative characters of order
  $r$ on $\gf_\varrho^*$ --- that is, all characters that are isomorphic to the
  trivial character on the subgroup $H$.  
  Thus
  \begin{align*}
    R_{f}^{(r)} & = \sum_{a\in\gf_\varrho^*}
    \left(\frac{1}{r}\sum_{\chi^r=\chi_{0}} \chi(f(a))\right) =
    \frac{1}{r}\sum_{\chi^r=\chi_{0}} \sum_{a\in\gf_\varrho^*}
    \chi(f(a))\\
    & \leq \frac{\varrho}{r}+ \frac{1}{r} \sum_{
      \substack{\chi^r=\chizero\\
        \chi\neq\chizero}} \left|\sum_{a\in\gf_\varrho} \chi(f(a))
    \right|.
    \end{align*}
    Here we use the obvious fact that 
    \[
    \sum_{a\in\gf_\varrho^*} \chizero(f(a)) \leq 
    \sum_{a\in\gf_\varrho} \chizero(f(a)) = \varrho-d\leq \varrho,
    \]
    where $d$ is the number of distinct roots of $f$ in $\gf_\varrho$.  We next
    employ the powerful theorem of \cite{Wei48} on character sums with
    polynomial arguments (see Theorem 5.41 of \cite{LidNie83}), which
    shows that if $f$ is \emph{not} a perfect $r$th power of another
    polynomial, and $\chi$ has order $r>1$, then
    \[
    \left| \sum_{a\in\gf_\varrho} \chi(f(a))\right| \leq (n-1) \varrho^{1/2}
    \leq \frac{\varrho}{2},
    \]
    using the fact that we insisted $n\leq 1+\sqrt{\varrho}/2$.
    Summing over the $r-1$ non-trivial characters of order $r$, we deduce that
    \[
    R_{f}^{(r)} \leq  \frac{\varrho}{r}+ \frac{r-1}{r}\cdot \frac{\varrho}{2} \leq
    \frac{3\varrho}{4}, 
    \]
    since $r\geq 2$.
\end{proof}

\subsection{Certifying specified powers over $\gf_q[x]$}

Theorem \ref{thm:fapow} allows us to detect when a polynomial
$f\in\gf_\varrho[x]$ is a perfect $r$th power, for known $r$
dividing $\varrho-1$:  choose random $\alpha\in\gf_\varrho$ and
evaluate $\xi=f(\alpha)^{(\varrho-1)/r}\in\gf_\varrho$.  Recall that
$\xi=1$ if and only if $f(\alpha)$ is an $r$th power.

\begin{list}{$\bullet$}{\itemsep=0pt\topsep=3pt}
\item If $f$ is an $r$th power, then clearly $f(\alpha)$ is an $r$th
  power and we always have $\xi=1$.
\item If $f$ is not an $r$th power, Theorem \ref{thm:fapow}
  demonstrates that for at least $1/4$ of the elements of $\gf_\varrho$,
  $f(\alpha)$ is not an $r$th power.  Thus, for $\alpha$ chosen
  randomly from $\gf_\varrho$ we would expect $\xi\neq 1$ with probability
  at least $1/4$.
\end{list}

For a polynomial $f\in\gf_q[z]$ over an arbitrary finite field
$\gf_q[x]$, $q-1$ is not necessarily divisible by $r$, so we will work
in a suitable extension. First, we can safely assume $r\ndivs q$ under
the requirement that the characteristic of $\gf_q$ is strictly greater
than $\deg f$, since in any case we must have $r \leq \deg f$. Then from
Fermat's Little Theorem, we know that $r\divs (q^{r-1}-1)$ and so we
construct an extension field $\gf_{q^{r-1}}$ over $\gf_q$ and proceed as
above.
We now present and analyze this more formally.

\begin{alg}[pprGF]{IsPerfectRthPowerGF}
\REQUIRE A prime power $q$, $f\in\gf_q[x]$ of degree $n$ such that
        $\characteristic(\gf_q) < n \leq 1+\sqrt{q}/2$,
        $r\in\NN$ a prime dividing $n$, and $\epsilon\in\RR_{>0}$
\ENSURE \TRUE\ if $f$ is the $r$th power of a polynomial in $\gf_q[x]$;
 \FALSE\ otherwise.

\STATE Find an irreducible $\Gamma\in\gf_q[z]$ of degree $r-1$,
successful with probability at least $\epsilon/2$

\STATE $\varrho \gets q^{r-1}$ 
\STATE Define
$\gf_{\varrho}=\gf_q[z]/(\Gamma)$

\STATE $m \gets 2.5 (1+\ceil{\log_2(1/\epsilon)})$

\FOR{$i$ from $1$ to $m$\label{pprGF:forA}}

\STATE Choose random $\alpha\in\gf_{\varrho}$

\STATE $\xi \gets f(\alpha)^{(\varrho-1)/r}\in\gf_{\varrho}$

\IF{ $\xi\neq 1$ } \RETURN \FALSE \ENDIF

\ENDFOR\label{pprGF:forB}

\RETURN \TRUE
\end{alg}

\noindent\textbf{Notes on \ref{pprGF}.}

To accomplish Step 1, a number of fast probabilistic methods are
available to find irreducible polynomials. We employ the algorithm of
\cite{Sho94}.  This algorithm requires $O((r^2\log r+r\log q)
\log r\log\log r)$ operations in $\gf_q$.  It is probabilistic of the
Las Vegas type, and we assume that it always stops within the number
of operations specified, and returns the correct answer with
probability at least $1/2$ and ``Fail'' otherwise (it never returns an
incorrect answer).  The algorithm is actually presented in
\cite{Sho94} as \emph{always} finding an irreducible polynomial, but
requiring \emph{expected} time as above; by not iterating indefinitely
our restatement allows for a Monte Carlo analysis in what follows.  To
obtain an irreducible $\Gamma$ with failure probability at most
$\epsilon/2$ we run (our modified) Shoup's algorithm
$1+\ceil{\log_2(1/\epsilon)}$ times.

The restriction that $n\leq 1+\sqrt{q}/2$ (or equivalently that $q\geq
4(n-1)^2$) is not at all limiting.  If this condition
is not met, simply extend $\gf_q$ with an extension of degree
$\nu=\ceil{\log_q(4(n-1)^2)}$ and perform the algorithm over
$\gf_{q^\nu}$.  At worst, each operation in $\gf_{q^\nu}$ requires
$O(\M(\log n))$ operations in $\gf_q$.

Here we define $\M(r)$ as a number of operations in $\F$ to
multiply two polynomials of degree $\leq r$ over $\F$, for any
field $\F$, or
the number of bit operations to multiply two integers with at most $r$
bits.  Using classical arithmetic $\M(r)$ is $O(r^2)$, while using the
fast algorithm of \cite{CanKal91} we may assume $\M(r)$ is $O(r\log
r\log\log r)$.

\begin{theorem}
  Let $q,f,n,r,\epsilon$ be as in the input to the algorithm \ref{pprGF}.
  If $f$ is a perfect $r$th power the
  algorithm always reports this.  If $f$ is not a perfect
  $r$th power then, on any invocation, this is
  reported correctly with probability at least~$1-\epsilon$.
\end{theorem}
\begin{proof}
  It is clear from the above discussion that the algorithm always
  works when $f$ is perfect power.  When $f$ is not a perfect power,
  each iteration of the loop will obtain $\xi\neq 1$ (and hence a
  correct output) with probability at least $1/4$.  By iterating the
  loop $m$ times we ensure that the probability of failure is at most
  $\epsilon/2$.  Adding this to the probability that Shoup's algorithm
  (for Step 1) fails yields a total probability of failure of at
  most~$\epsilon$.
\end{proof}

\begin{theorem}
  \label{thm:perfectRthGF}
  On inputs as specified, the algorithm 
  \ref{pprGF} requires $O( (r\M(r)\log r\log q) \cdot\log
  (1/\epsilon))$ operations in $\gf_q$ plus the cost to evaluate
  $\alpha\mapsto f(\alpha)$ at $O(\log(1/\epsilon))$ points
  $\alpha\in\gf_{q^{r-1}}$.
\end{theorem}
\begin{proof}
  As noted above,
  each iteration through the algorithm of \citet{Sho94}
  requires $O((r^2\log r+r\log q) \log r\log\log r)$ field operations,
  which is within the time specified.  The main cost of
  the loop in Steps \ref{pprGF:forA}--\ref{pprGF:forB} 
  is computing $f(\alpha)^{(\varrho-1)/r}$,
  which requires $O(\log \varrho)$ or $O(r\log q)$ operations in
  $\gf_{\varrho}$ using repeated squaring, plus one evaluation of $f$
  at a point in $\gf_{\varrho}$.  Each operation in $\gf_{\varrho}$
  requires $O(\M(r))$ operations in $\gf_q$, and we repeat the loop
  $O(\log(1/\epsilon))$ times.
\end{proof}

\begin{corollary} \label{cor:pprGF}
  Given $f\in\gf_q[x]$ of degree $n$ with $\tms{f}=t$, 
  and $r\in\NN$ a prime dividing $n$, we can determine if
  $f$ is an $r$th power with 
  \[
  O\left( \left(r\M(r)\log r\log q + t\M(r)\log n\right) 
    \cdot\log (1/\epsilon)\right)
  \]
  operations in $\gf_q$, provided $n > \characteristic(\gf_q)$. 
  When $f$ is an $r$th
  power, the output is always correct, while if $f$ is not an $r$th
  power, the output is correct with probability at least~$1-\epsilon$.
\end{corollary}

\subsection{Certifying specified powers over $\ZZ[x]$}

For an integer polynomial $f\in\ZZ[x]$, we proceed by working in the
homomorphic image of $\ZZ$ in $\gf_p$ (and then in an extension of
that field).  We must ensure that the homomorphism preserves the
perfect power property we are interested in with high probability.
For any polynomial $g\in\F[x]$, let $\disc(g)=\res(g,g')$ be the
discriminant of $g$ (the resultant of $g$ and its first derivative).
It is well known that $g$ is squarefree if and only if $\disc(g)\neq
0$.  Also define $\lcoeff(g)$ as the leading coefficient of $g$, the
coefficient of the highest power of $x$ in $g$.
Finally, for $g\in\ZZ[x]$ and $p$ a prime, denote by $g \rem p$ the
unique polynomial in $\gf_p[x]$ with all coefficients in $g$ reduced
modulo $p$.

\begin{lemma}\label{lem:ZtoGF}
  Let $f\in\ZZ[x]$ and $\ftil=f/\gcd(f,f')$ its squarefree part.  Let
  $p$ be a prime such that $p\ndivs\disc(\ftil)$ and
  $p\ndivs\lcoeff(f)$.  Then $f$ is a perfect power in $\ZZ[x]$ if and
  only if $f\rem p$ is a perfect power in $\gf_p[x]$.
\end{lemma}
\begin{proof}
  Clearly if $f$ is a perfect power, then $f\rem p$ is a perfect
  power in $\ZZ[x]$.  To show the converse, assume that
  $f=f_1^{s_1}\cdots f_m^{s_m}$ for distinct irreducible
  $f_1,\ldots,f_m\in\ZZ[x]$, so $\ftil=f_1\cdots f_m$.  Clearly
  $f\equiv f_1^{s_1}\cdots f_m^{s_m}\bmod p$ as well, and because
  $p\ndivs\lcoeff(f)$ we know $\deg (f_i\rem p) = \deg f_i$ for
  $1\leq i\leq m$.  Since $p\ndivs\disc(\ftil)$, $\ftil\rem p$ is
  squarefree (see \cite{GatGer03}, Lemma 14.1), and each of the
  $f_i\rem p$ must be pairwise relatively prime and squarefree
  for $1\leq i\leq m$.  Now suppose $f\rem p$ is a perfect $r$th
  power modulo $p$.  Then we must have $r\divs s_i$ for $1\leq i\leq
  m$.  But this immediately implies that $f$ is a perfect power in
  $\ZZ[x]$ as well.
\end{proof}

Given any polynomial $g=g_0+g_1x+\cdots+g_mx^m\in\ZZ[x]$, we define
the height or coefficient $\infty$-norm of $g$ as
$\inorm{g}=\max_i|g_i|$.  Similarly, we define the coefficient 1-norm
of $g$ as $\norm{g}_1=\sum_i |g_i|$, and 2-norm as
$\norm{g}_2=\left(\sum_i |g_i|^2\right)^{1/2}$.  
With $f,\ftil$ as in Lemma~\ref{lem:ZtoGF}, $\ftil$ divides
$f$, so we can employ the factor bound of \cite{Mig74} to obtain
\[
\inorm{\ftil}\leq 2^n\norm{f}_2\leq 2^n\sqrt{n+1}\cdot\norm{f}_\infty.
\]
Since $\disc(\ftil)=\res(\ftil,\ftil')$ is the determinant of matrix
of size at most $(2n-1)\times (2n-1)$, Hadamard's inequality implies
\[
|\disc(\ftil)|\leq \left(2^n\left(n+1\right)^{1/2}\inorm{f}\right)^{n-1}
\left(2^n\left(n+1\right)^{3/2}\inorm{f}\right)^n 
< 2^{2n^2}(n+1)^{2n}\cdot\inorm{f}^{2n}.
\]
Also observe that $|\lcoeff(f)|\leq\inorm{f}$.  Thus, 
the product $\disc(\ftil)\cdot\lcoeff(f)$ has at most
\[
\mu=\left\lceil\frac{\left\lceil
    \log_2\left(2^{2n^2}\left(n+1\right)^{2n}\inorm{f}^{2n+1}\right)
  \right\rceil}
  {\left\lfloor\log_2\left(4\left(n-1\right)^2\right)\right\rfloor}\right\rceil
\]
prime factors greater than $4(n-1)^2$ (we require the lower bound
$4(n-1)^2$ to employ Theorem \ref{thm:fapow} without resorting to
field extensions).  Choose an integer $\gamma\geq 4(n-1)^2$ such that the
number of primes between $\gamma$ and
$2\gamma$ is at least $4\mu+1$.  By \cite{RosSch62}, Corollary 3,
the number of primes in this range
is at least $3\gamma/(5\ln\gamma)$ for $\gamma\geq 21$.  

Now let $\gamma \geq \max\{21\mu\ln\mu,226\}$. It is easily confirmed
that if $\mu \leq 6$ and $\gamma \geq 226$, 
then $3\gamma/(5\ln\gamma) > 4\mu + 1$.
Otherwise, if $\mu \geq 7$, then $\ln(21\ln \mu) < 2\ln\mu$, so
\[\frac{\gamma}{\ln\gamma} \geq 
\frac{21\mu\ln\mu}{\ln\mu + \ln(21\ln\mu)} >
7\mu,\]
and therefore $3\gamma/(5\ln\gamma) > 21\mu/5 > 4\mu + 1$.

Thus, if $\gamma \geq \max\{21\mu \ln\mu,226\}$, then a
random prime not equal to $r$ in the range $\gamma \ldots 2\gamma$
divides $\lcoeff(f)\cdot\disc(f)$ with probability at most $1/4$.
Primes $p$ of this size have only $\log_2 p \in O(\log
n+\log\log\inorm{f})$ bits.

\begin{alg}[pprZ]{IsPerfectRthPowerZ}
\REQUIRE $f\in\ZZ[x]$ of degree $n$;
 $r\in\NN$ a prime dividing $n$;
 $\epsilon\in\RR_{>0}$;
\ENSURE \TRUE\ if $f$ is the $r$th power of a polynomial in $\ZZ[x]$;
 \FALSE\ otherwise

\STATE $\mu\gets \left\lceil\left\lceil
    \log_2\left(2^{2n^2}\left(n+1\right)^{2n}\inorm{f}^{2n+1}\right)
  \right\rceil /
  \left\lfloor\log_2\left(4\left(n-1\right)^2\right)\right\rfloor\right\rceil$
\STATE $\gamma \gets \max\{\ceil{21\mu \ln\mu},4(n-1)^2,226\}$
\FOR{ $i$ from 1 to $\ceil{\log_2(1/\epsilon)}$\label{pprZ:loopbegin}} 
\STATE $p \gets $ random prime in the range $\gamma \ldots 2\gamma$
\IF  {NOT \ref{pprGF}($p$, $f\rem p$, $r$, $1/4$)\label{pprZ:compGF}}
\RETURN \FALSE
\ENDIF
\ENDFOR \label{pprZ:loopend}
\RETURN \TRUE
\end{alg}

\begin{theorem}
  Let $f\in\ZZ[x]$ of degree $n$, $r\in\NN$ dividing $n$ and
  $\epsilon\in\RR_{>0}$.  If $f$ is a perfect $r$th power, the algorithm
  \ref{pprZ} always reports this. If $f$ is not a perfect $r$th power,
  on any invocation of the algorithm, this is reported correctly with
  probability at least $1-\epsilon$.
\end{theorem}
\begin{proof}
  If $f$ is an $r$th power then so is $f\rem p$ for any prime $p$,
  and so is any $f(\alpha)\in\gf_p$.  Thus, the algorithm always
  reports that $f$ is an $r$th power.  Now suppose $f$ is not an $r$th
  power.  If $p\divs \disc(f)$ or $p\divs \lcoeff(f)$ 
  it may happen that $f\rem p$ \emph{is}
  an $r$th power. This happens with probability at most $1/4$ and we
  will assume that the worst happens in this case.  When
  $p\ndivs\disc(f)$ and $p\ndivs\lcoeff(f)$, 
  the probability that \ref{pprGF} incorrectly
  reports that $f$
  is an $r$th power is also at most $1/4$, by our choice of parameter
  $\epsilon$ in the call to \ref{pprGF}.  Thus, on any iteration of steps
  \ref{pprZ:loopbegin}--\ref{pprZ:loopend}, the probability of
  finding that $f$ is an $r$th power is at most $1/2$.  The
  probability of this happening $\ceil{\log_2(1/\epsilon)}$ times is
  at most $\epsilon$.
\end{proof}

\begin{theorem}
  On inputs as specified, the algorithm
  \ref{pprZ} requires
  \[
  O\Bigl(r\M(r)\log r \cdot \M(\log n+\log\log\inorm{f})
  \cdot (\log n+\log\log\inorm{f}) \cdot\log(1/\epsilon)  \Bigr),
  \]
  or $\softO(r^2(\log n+\log\log\inorm{f})^2\cdot\log(1/\epsilon))$
  bit operations, plus the cost to evaluate $(\alpha,p)\mapsto
  f(\alpha)\bmod p$ at $O(\log(1/\epsilon))$ points $\alpha\in\gf_p$
  for primes $p$ with $\log p \in O(\log n+\log\log\inorm{f})$.
\end{theorem}
\begin{proof}
  The number of operations required by each iteration is dominated by
  Step \ref{pprZ:compGF}, 
  for which $O(r\M(r)\log r \log p)$ operations in $\gf_p$ is
  sufficient by Theorem \ref{thm:perfectRthGF}.  Since $\log p\in O(\log
  n+\log\log\inorm{f})$ we obtain the final complexity as stated.
\end{proof}

We obtain the following corollary for $t$-sparse polynomials in
$\ZZ[x]$.  This follows since the cost in bit operations
of evaluating a $t$-sparse
polynomial $f\in\ZZ[x]$ modulo a prime $p$ is $O(t\log\inorm{f}\log
p+t\log n\M(\log p))$.

\begin{corollary}\label{cor:pprZ}
  Given $f\in\ZZ[x]$ of degree $n$, with $\tms{f}=t$,
  and $r\in\NN$ a prime dividing $n$, we can determine if
  $f$ is an $r$th power with 
  \[
  \softO\left( (r^2 \log^2 n + t\log^2 n
  +t \log\inorm{f}\log n) \cdot\log (1/\epsilon)
  \right)
  \] 
  bit operations. When $f$ is an $r$th power, the output is always
  correct, while if $f$ is not an $r$th power, the output is correct
  with probability at least $1-\epsilon$.
\end{corollary}

\subsection{An upper bound on $r$.}

In this subsection we show that if $f=h^r$ and $f\neq x^n$ then
$r$ must be small. Over $\ZZ[x]$ we show that $\norm{h}_2$ is small as
well.  A sufficiently strong result over many fields is demonstrated
by \cite{Sch87}, Theorem 1, where it is shown that if $f$ has sparsity
$t\geq 2$ then $t\geq r+1$ (in fact a stronger result is shown
involving the sparsity of $h$ as well).  This holds when either the
characteristic of the ground field of $f$ is zero or greater than
$\deg f$.

Here we give a (much) simpler result for polynomials in $\ZZ[x]$, which
bounds $\norm{h}_2$ and is stronger at least in its dependency on $t$
though it also depends upon the coefficients of $f$.

\begin{theorem}
  \label{thm:Ztbound}
  Suppose $f\in\ZZ[x]$ with $\deg f = n$ and $\tms{f}=t$, and $f=h^r$
  for some $h\in\ZZ[x]$ of degree $s$ and $r\geq 2$.  Then $\norm{h}_2\leq
  \norm{f}_1^{1/r}$.
\end{theorem}
\begin{proof*}
  Let $p>n$ be prime and $\zeta\in\CC$ a $p$th primitive root of
  unity. Then
  \[
  \norm{h}_2^2=\sum_{0\leq i\leq s}|h_i|^2
  = \frac{1}{p}\sum_{0\leq i<p} |h(\zeta^i)|^2.
  \]
  (this follows from the fact that the Discrete Fourier Transform
  (DFT) matrix is orthogonal).  In other words, the average value of
  $|h(\zeta^i)|^2$ for $i=0\ldots p-1$ is $\norm{h}^2_2$, and so there
  exists a $k\in\{0,\ldots,p-1\}$ with $|h(\zeta^k)|^2\geq
  \norm{h}_2^2$.  Let $\theta=\zeta^k$.  Then clearly
  $|h(\theta)|\geq\norm{h}_2$.  We also note that
  $f(\theta)=h(\theta)^r$ and $|f(\theta)|\leq \norm{f}_1$,
  since $|\theta|=1$. Thus,
  \[
  \norm{h}_2\leq |h(\theta)|=|f(\theta)|^{1/r}\leq
  \norm{f}_1^{1/r}. \qed
  \]
\end{proof*}
  
\smallskip
\noindent
The following corollary is particularly useful.

\begin{corollary} 
  \label{upper-r-pp} 
  If $f \in \ZZ[x]$ is not of the form $x^n$, and $f = h^r$ for some
  $h \in \ZZ[x]$, then
  \begin{list}{}{\labelwidth=15pt\itemsep=0pt\topsep=2pt}
  \item[(i)] $r\leq 2\log_2\norm{f}_1$,
  \item[(ii)] $\tms{h}\leq \norm{f}_1^{2/r}$.
  \end{list}
\end{corollary}
\begin{proof}
  Part (i) follows since $\norm{h}_2 \geq \sqrt{2}$.
  Part (ii) follows because $\norm{h}_2\geq \sqrt{\tms{h}}$.
\end{proof}

These bounds relate to the sparsity of $f$ since 
$\norm{f}_1 \leq \tms{f}\inorm{f}$.

\subsection{Perfect power detection algorithm}
We can now complete the perfect power detection algorithm, when we are
given only the $t$-sparse polynomial $f$ (and not $r$).

\begin{alg}[ppZ]{IsPerfectPowerZ}
\REQUIRE $f\in\ZZ[x]$ of degree $n$ and sparsity $t \geq 2$, 
 $\epsilon\in\RR_{>0}$
\ENSURE \TRUE\ and $r$ if $f=h^r$ for some $h\in\ZZ[x]$;
 \FALSE\ otherwise.

\STATE $\calP \gets \{ \mbox{primes}~ r\divs n ~\mbox{and}~ r \leq
2\log_2(t\inorm{f}) \}$

\FOR{ $r\in \calP$  }

\IF { \ref{pprZ}($f$, $r$, $\epsilon/\#\calP$)  }
        
\RETURN \TRUE\ and $r$
\ENDIF
\ENDFOR
\RETURN \FALSE
\end{alg}

\begin{theorem} \label{thm:ppZ}
  If $f\in\ZZ[x]=h^r$ for
  some $h\in\ZZ[x]$, the algorithm \ref{ppZ}
  always returns ``True'' and returns $r$
  correctly with probability at least $1-\epsilon$.  
  Otherwise, it returns ``False'' with probability at least
  $1-\epsilon$.\\ 
  The algorithm requires
  $\softO(t\log^2\inorm{f}\log^2n\log(1/\epsilon))$ bit operations.
\end{theorem}

\begin{proof}
From the preceding discussions, we can see that if $f$ is a perfect power,
then it must be a perfect $r$th power for some $r \in \calP$. So the algorithm
must return true on some iteration of the loop. However, it may incorrectly
return true \emph{too early} for an $r$ such that $f$ is not actually an $r$th
power; the probability of this occurring is the probability of error when $f$
is not a perfect power, and is less than $\epsilon/\#\calP$ at each
iteration. So the probability of error on any iteration is at most $\epsilon$,
which is what we wanted.

The complexity result follows from the fact that
each $r \in O(\log t + \log \norm{f}_\infty )$ and using
Corollary \ref{cor:pprZ}.
\end{proof}

For polynomials in $\gf_q[x]$ we use Schinzel's bound that $r\leq t-1$
and obtain the following algorithm.

\begin{alg}[ppGF]{IsPerfectPowerGF}
\REQUIRE A prime power $q$, $f\in\gf_q[x]$ of degree $n$ and sparsity $t$
  such that $n < \characteristic(\gf_q)$, and $\epsilon\in\RR_{>0}$
\ENSURE \TRUE\ and $r$ if $f=h^r$ for some $h\in\gf_q[x]$;
 \FALSE\ otherwise.

\STATE $\calP \gets \{ \mbox{primes}~ r\divs n ~\mbox{and}~ r \leq t \}$

\FOR{ $p\in \calP$  }

\IF { \ref{pprGF}($f$, $r$, $\epsilon/\#\calP$) }
        
\RETURN \TRUE\ and $r$
\ENDIF
\ENDFOR
\RETURN \FALSE
\end{alg}

\begin{theorem}
  If $f=h^r$ for
  $h\in\gf_q[x]$, the algorithm \ref{ppGF} 
  always returns ``True'' and returns $r$
  correctly with probability at least $1-\epsilon$.  
  Otherwise, it returns ``False'' with probability at least
  $1-\epsilon$.  The algorithm requires $\softO(t^3(\log
  q+\log n))$ operations in~$\gf_q$.
\end{theorem}

\begin{proof}
The proof is equivalent to that of Theorem \ref{thm:ppZ}, using the
complexity bounds in Corollary \ref{cor:pprGF}.
\end{proof}

\subsection{Detecting multivariate perfect powers}

In this subsection we examine the problem of detecting multivariate
perfect powers.  That is, given a lacunary $f\in\F[x_1,\ldots,x_\ell]$
of total degree $n$ as in \eqref{eq:lacpolymv}, we want to determine if
$f=h^r$ for some $h\in\F[x_1,\ldots,x_\ell]$ and $r\in\NN$.  This is
done simply as a reduction to the univariate case.

First, given $f\in\F[x_1,\ldots,x_\ell]$, define the squarefree part
$\ftil\in\F[x_1,\ldots,x_\ell]$ as the squarefree polynomial of
highest total degree which divides $f$.

\begin{lemma}
  Let $f\in\F[x_1,\ldots,x_\ell]$ be of total degree $n>0$ and let
  $\ftil\in\F[x_1,\ldots,x_\ell]$ be the squarefree part of
  $f$. Define
  \[
  \Delta=\disc_x(\ftil(y_1x,\ldots,y_\ell x))
           =\res_x(\ftil(y_1x,\ldots,y_\ell x),\ftil'(y_1x,\ldots,y_\ell x))
           \in\F[y_1,\ldots,y_\ell]
  \]
  and
  \[
  \Lambda=\lcoeff_x(f(y_1x,\ldots,y_\ell x))
  \in\F[y_1,\ldots,y_\ell]
  \]
  for independent indeterminates $x,y_1,\ldots,y_\ell$.  Assume that
  $a_1,\ldots,a_\ell\in\F$ with 
  $$\Delta(a_1,\ldots,a_\ell)\neq 0 \qquad\text{and}\qquad
  \Lambda(a_1,\ldots,a_n)\neq 0.$$  Then $f(x_1,\ldots,x_\ell)$ is a
  perfect power if and only if $f(a_1x,\ldots,a_\ell x)$
  $\in\F[x]$ is a perfect power.
\end{lemma}
\begin{proof}
  Clearly if $f$ is a perfect power, then $f(a_1 x,\ldots,a_\ell x)$
  is a perfect power.  To prove the converse, assume that
  \[
  f=f_1^{s_1}f_2^{s_2}\cdots f_m^{s_m}
  \]
  for irreducible $f_1,\ldots,f_m\in\F[x_1,\ldots,x_\ell]$.  Then
  \[
  f(y_1x,\ldots,y_mx) = 
  f_1(y_1x,\ldots,y_mx)^{s_1}\cdots f_m(y_1x,\ldots,y_mx)^{s_m}
  \]
  and each of the $f_i(y_1x,\ldots,y_mx)$ are irreducible.  Now, since
  $\Lambda(a_1,\ldots,a_m)\neq 0$, we know the
  $\deg(f(a_1x,\ldots,a_\ell x))=\deg f$ (the total degree of $f$).
  Thus, $\deg f_i(a_1x,\ldots,a_\ell x)=\deg f_i$ for $1\leq
  i\leq\ell$ as well.  Also, by our assumption,
  $\disc(f(a_1x,\ldots,a_\ell x))\neq 0$, so all of the
  $f_i(a_1x,\ldots,a_\ell x)$ are squarefree and pairwise relatively
  prime for $1\leq i\leq k$, and
  \[
  f(a_1x,\ldots,a_\ell x)=f_1(a_1 x,\ldots, a_\ell x)^{s_1}\cdots
      f_m(a_1 x,\ldots,a_\ell x)^{s_m}. 
  \]
  Assume now that $f(a_1 x,\ldots,a_\ell x)$ is an $r$th perfect
  power.  Then $r$ divides $s_i$ for $1\leq i\leq m$.  This
  immediately implies that $f$ itself is an $r$th perfect power.
\end{proof}

It is easy to see that the total degree of $\Delta$ is less than
$2n^2$ and the total degree of $\Lambda$ is less than $n$, and that
both $\Delta$ and $\Lambda$ are non-zero.  Thus, for randomly chosen
$a_1,\ldots,a_\ell$ from a set $\calS\subseteq\F$ of size at least
$8n^2+4n$ we have $\Delta(a_1,\ldots,a_\ell)=0$ or
$\Lambda(a_1,\ldots,a_\ell)=0$ with probability less than $1/4$, by
\cite{Zip79} or \cite{Sch80}.  This can be made arbitrarily small by
increasing the set size and/or by repetition.  We then run the
appropriate univariate algorithm over $\F[x]$ 
to identify whether or not $f$ is a perfect power, and if so,
to find $r$. Note that, for integer polynomials,
$f(a_1x,\ldots,a_\ell x)$ need not be explicitly computed over $\ZZ[x]$;
this can be delayed until a finite field is chosen in the
\ref{pprZ} algorithm, in order to preserve polynomial time.

\section{Computing perfect roots}
\label{sec:computepp}

Once we have determined that $f \in \F[x]$ is equal to $h^r$ for some
$h \in \F[x]$, the next task is to actually compute $h$.
Unfortunately, as noted in the introduction, there are no known bounds
on $\tau(h)$ which are polynomial in $\tau(f)$.

The question of how sparse the polynomial root of a sparse polynomial
must be (or equivalently, how dense any power of a dense polynomial
must be) relates to some questions first raised by \citet{Erd49} on
the number of terms in the square of a polynomial.  Schinzel extended
this work to the case of perfect powers and proved that $\tau(h^r)$
tends to infinity as $\tau(h)$ tends to infinity \citep{Sch87}. Some
conjectures of Schinzel suggest that $\tau(h)$ should be
$O(\tau(f))$. A recent breakthrough of \citet{Zan07Acta} shows that
$\tau(h)$ is bounded by a function which does not depend on $\deg f$,
but this bound is unfortunately not polynomial in $\tau(f)$.

Our own (limited) investigations, along with more extensive
ones by \cite{CopDav91}, and \cite{Abb02}, suggest that, for any
$h\in\F[x]$, where the characteristic of $\F$ is not too small,
$\tms{h} \in O(\tms{h^r}+r)$. We skirt this problem in two ways: our
first algorithm is output-sensitive, and the second relies on a modest
conjecture.

\subsection{Computing $r$th roots in polynomial-time (without conditions)}

In this subsection we present an algorithm for computing an $h$ such
that $f=h^r$ given $f\in\ZZ[x]$ and $r\in\ZZ$ or showing that no such
$h$ exists.  The algorithm is deterministic and requires time
polynomial in $t=\tau(f)$, $\log\deg f$, $\log\inorm{f}$ and a given
upper bound $\mu$ on $m=\tau(h)$.  Neither its correctness nor
complexity is conditional on any conjectures.  We will only
demonstrate that this algorithm requires polynomial time. A more
detailed analysis is performed on the (more efficient) algorithm of
the next subsection (though that complexity is subject to a modest
conjecture).

The basic idea of the algorithm here is that we can recover all the
coefficients in $\QQ$ as well as modular information about the
exponents of $h$ from a homomorphism into a small cyclotomic field over
$\QQ$. Doing this for a relatively small number of cyclotomic fields
yields $h$.

Assume that (the unknown) $h\in\ZZ[x]$ has form 
\[
h = \sum_{1\leq i\leq m} b_{i} x^{d_i} ~~~\mbox{for $b_1,\ldots,b_m\in\ZZ\nonzero$,
and $0\leq d_1<d_2<\cdots < d_m,$} 
\]
and that $p>2$ is a prime distinct from $r$  such that 
\begin{equation}
\label{eq:primecond}
p\ndivs \prod_{1\leq i<j\leq m} (d_j-d_i),
~~\mbox{and}~~
p\ndivs \prod_{1\leq i\leq m} (d_i+1).
\end{equation}
Let $\zeta_{p}\in\CC$ be a $p$th primitive root of unity, and
$\Phi_p=1+z+\cdots+z^{p-1}\in\ZZ[z]$ its minimal polynomial, the $p$th
cyclotomic polynomial (irreducible in $\QQ[z]$).
Computationally we represent $\QQ(\zeta_p)$ as $\QQ[z]/(\Phi_p)$, with
$\zeta_p\equiv z\bmod\Phi_p$.  Observe that $\zeta_{p}^k=\zeta_{p}^{k
  \rem p}$ for any $k\in\ZZ$, where $k \rem p$ is the least
non-negative residue of $k$ modulo $p$.  Thus
\[
h(\zeta_{p})=h_p(\zeta_{p})~~~
\mbox{for}~~~ h_p=\sum_{1\leq i\leq m} b_i x^{d_i\rem p}\in\ZZ[x],
\]
and $h_p$ is the unique representation of $h(\zeta_{p})$ as a
polynomial of degree less than $p-1$.  This follows from the
conditions \eqref{eq:primecond} on our choice of prime $p$ because
\begin{itemize}
\item No pair of distinct exponents $d_i$ and $d_j$ of $h$ is equivalent modulo $p$ (since $p\ndivs
  (d_i-d_j)$);
\item All the exponents reduced modulo $p$ are strictly less than
  $p-1$ (since our conditions imply $d_i\nequiv (p-1)\bmod p$ for
  $1\leq i\leq m$).
\end{itemize}
This also implies that the coefficients of $h_p$ are exactly the same
as those of $h$, albeit in a different order.

Now observe that we can determine $h_p$ quite easily from the roots of
\[
\Gamma_p(y)=y^r-f(\zeta_{p})\in\QQ(\zeta_{p})[y].
\] 
These roots can be found by factoring the polynomial $\Gamma_p(y)$ in
$\QQ(\zeta_{p})[y]$, and the roots in $\CC$ must be
$\omega^ih(\zeta_{p})\in\CC$ for $0\leq i<r$, where $\omega$ is a
primitive $r$th root of unity. When $r>2$, and since 
$p\neq r$,
the only $r$th root of unity in $\QQ(\zeta_{p})$ is
$1$. Thus $\Gamma_p(y)$ has 
one linear factor, $y-h(\zeta_p)=y-h_p(\zeta_p)$,
precisely determining $h_p$.
When $r=2$, we have
\[
\Gamma_{p}(y)=(y-h(\zeta_p))(y+h(\zeta_p))=(y-h_p(\zeta_p))(y+h_p(\zeta_p))
\]
and we can only determine $h_p(\zeta_p)$ (and $h_p$ and, for that
matter, $h$) up to a factor of $\pm 1$. However, the exponents of
$h_p$ and $-h_p$ are the same, and the ambiguity is only in the
coefficients (which we resolve later).

Finally, we 
perform the above operations for a sequence of
cyclotomic fields $\QQ(\zeta_{p_1})$, $\QQ(\zeta_{p_2})$, \ldots,
$\QQ(\zeta_{p_k})$ such that the primes in $\calP=\{p_1,\ldots,p_k\}$ allow us to recover
all the exponents in $h$.   Each prime 
gives the set of exponents of $h$ reduced modulo that prime, and
\emph{all} the coefficients of $h$ in $\ZZ$. That is, from each computation
with $p\in\calP$ we obtain 
\[
\calC=\left\{ b_1,\ldots,b_m \right\}~~~\mbox{and}~~ 
\calE_p=\left\{ d_1\rem p, d_2\rem p,\ldots, d\rem p\right\},
\]
but with no clear information about the order of these sets.  
In particular, it is not obvious
how to correlate the exponents modulo the different primes. To do
this we employ the clever sparse interpolation technique of
\cite{Garg-Schost:2008} (based on a method of \citet{GriKar87} for a
different problem), which interpolates the symmetric polynomial in the exponents:
\[
g=(x-d_1)(x-d_2)\cdots (x-d_m)\in\ZZ[x].
\]
For each $p\in\calP$ we compute the symmetric polynomial
modulo $p$,
\[
g_{p}=(x-(d_1\rem p))(x-(d_2\rem p))\cdots (x-(d_m\rem p))
\equiv g\bmod p,
\]
for which we do not need to know the order of the exponent residues.
We then determine $g\in\ZZ[x]$ by the Chinese remainder theorem and
factor $g$ over $\ZZ[x]$ to find the $d_1,\ldots,d_m\in\ZZ$.  Thus
the product of all primes in $p\in\calP$ must be at least $2\inorm{g}$
to recover the coefficients of $g$ uniquely.  It is easily seen that
$2\inorm{g}\leq 2n^m$.

As noted above, the computation with each $p\in\calP$ recovers all the
exponents of $h$ in $\ZZ$, so using only one prime $p\in\calP$, we
determine the $j$th exponent of $h$ as the coefficient of $x^{d_j\rem
  p}$ in $h_p$ for $1\leq j\leq m$. If $r=2$ we can choose either of
the roots of $\Gamma_{p}(y)$ (they differ by only a sign) to recover
the coefficients of $h$.

Finally, 
we certify that $f=h^r$ by
taking logarithmic derivatives to obtain 
$$\frac{f'}{f} = \frac{rh'h^{r-1}}{h^r},$$ which simplifies to
$f'h = rh'f$. 
This relation is easily confirmed in polynomial time,
and along with checking leading
coefficients implies that in fact $f=h^r$.

The above discussion is summarized in the following algorithm.

\begin{alg}[alg:algebraic]{ComputeRootAlgebraic}
\REQUIRE $f \in \ZZ[x]$ as in \eqref{eq:lacpolyuni} with $\deg f = n$,
and $r,\mu \in \NN$
\ENSURE $h \in \ZZ[x]$ such that $f=h^r$ and $\tms{h}\leq \mu$, provided such
an $h$ exists

\STATE $\gamma \gets$ smallest integer $\geq 21$ such that
	$3\gamma/(5\ln\gamma)\geq (\mu^2+2\mu)\log_2n$ \label{setgamma}

\STATE $\calP \gets $ $\{p\in \{\gamma,\ldots,2\gamma\}~\mbox{and}~ p~\mbox{prime}\}$ \label{setP}

\FOR{$p\in\calP$}

\STATE Represent $\QQ(\zeta_p)$ by 
	$\QQ[x]/(\Phi_p)$, where $\Phi_p \gets 1+z+\cdots+z^{p-1}$ 
	and $\zeta_p \equiv z \bmod \Phi_p$
	
\STATE Compute $f(\zeta_{p})=\sum_{1\leq i\leq t}
c_i\zeta_{p}^{e_i\rem p}\in\ZZ[\zeta_p]$

\STATE Factor $\Gamma_p(y)\gets y^r-f(\zeta_p)\in\QQ(\zeta_p)[y]$ over
$\QQ(\zeta_p)[y]$ \label{factorGamma}

\IF{$\Gamma_p(y)$ has no roots in $\ZZ[\zeta_p]$} \RETURN ``$f$ is not
an $r$th power of a $\mu$-sparse polynomial''
\ENDIF

\STATE Let $h_p(\zeta_p)\in\ZZ[\zeta_p]$ be a root of $\Gamma_p(y)$\\
Write $h_p(x)=\sum_{1\leq i \leq m_p} b_{ip} x^{d_{ip}}$,
for $b_{ip}\in\ZZ$ and distinct $d_{ip}\in\NN$ for $1\leq i\leq m_p$

\IF{$\deg h_p = p-1$}
\STATE $m_p\gets 0$; Continue with next prime $p\in\calP$ at Step 3
  \label{checkdeg}
\ENDIF

\STATE $g_{p} \gets (x-d_{1p})(x-d_{2p})\cdots (x-d_{m_p p})\in\ZZ_{p}[x]$

\ENDFOR

\STATE $m\gets\max\{ m_p:\ p\in\calP \}$

\STATE $\calP_0\gets \{p\in\calP:\ m_p=m\}$ \label{setP0}

\STATE Reconstruct $g\in\ZZ[x]$ from $\{g_p\}_{p\in\calP_0}$ by the Chinese
Remainder Algorithm \label{setg}

\STATE $\{d_1,d_2,\ldots,d_k\} \gets $ distinct integer roots of $g$

\IF{ $k<m$}
\RETURN ``$f$ is not an $r$th power of a $\mu$-sparse polynomial''
\ENDIF

\STATE Choose any $p\in\calP_0$. For $1\leq j\leq m$, let $b_j\in\ZZ$ be
the coefficient of $x^{d_j\rem p}$ in $h_p$

\STATE $h \gets \sum_{1\leq j\leq m} b_jx^{d_j}$

\IF{$f'h = rh'f$ \textbf{and} $\lc(f) = \lc(h)^r$\label{finalcheckA}}
  \RETURN $h$
\ELSE
  \RETURN ``$f$ is not an $r$th power of a $\mu$-sparse polynomial''
\ENDIF \label{finalcheckB}
\end{alg}

\begin{theorem}
  The algorithm \texttt{ComputeRootAlgebraic} works as stated.  It
  requires a number of bit operations polynomial in $t=\tau(f)$,
  $\log\deg f$, $\log\norm{f}$, and $\mu$.
\end{theorem}
\begin{proof}
  We assume throughout the proof that there \emph{does} exist an
  $h\in\ZZ[x]$ such that $f=h^r$ and $\tau(h)\leq\mu$.  If it does
  not, this will be caught in the test in 
  Steps~\ref{finalcheckA}--\ref{finalcheckB} by the above
  discussion, if not before.
  
  In Steps~\ref{setgamma}--\ref{setP} 
  we construct a set of primes $\calP$ which is
  guaranteed to contain sufficiently many \emph{good} primes to
  recover $g$, where primes are good in the sense that for all
  $p\in\calP$
  \[
  \beta=r\cdot \prod_{1\leq i<j\leq m} (d_j-d_i) \cdot \prod_{1\leq
    i\leq m} (d_i+1) \ \nequiv\  0\bmod p.
  \]
  It is easily derived that $\beta < n^{\mu^2}$, which has fewer than
  $\log_2\beta\leq \mu^2\log_2n$ prime factors, so there are at most
  $\mu^2\log_2n$ \emph{bad} primes.  We also need to recover $g$ in
  Step~\ref{setg}, and $\inorm{g}\leq n^\mu$, for which we need at least
  $1+\log_2\norm{g}\leq 2\mu\log_2n$ good primes.  Thus if $\calP$ has
  at least $(\mu^2+2\mu)\log_2 n$ primes, there are a sufficient number of
  good primes to reconstruct $g$ in Step~\ref{setg}.

  By \cite{RosSch62}, Corollary 3, for $\gamma\geq 21$ we
  have that the number of primes in $\{\gamma,\ldots,2\gamma\}$ is at
  least $3\gamma/(5\ln\gamma)$, which is at least $(\mu^2+2\mu)\log_2n$ by our
  choice of $\gamma$ in Step~\ref{setgamma}, and $\gamma\in
  \softO(\mu^2\log(n))$. Numbers of this size can easily be
  tested for primality.

  Since we assume that a root $h$ exists, $\Gamma_p(y)$ will always
  have exactly one root $h_p\in\ZZ[\zeta_p]$ when $r>2$, and exactly
  two roots in $\ZZ[\zeta_p]$ when $r=2$ (differing only by sign).

  Two conditions cause the primes to be identified as bad.  If the map
  $h\mapsto h(\zeta_p)$ causes some exponents of $h$ to collide modulo
  $p$, this can only reduce the number of non-zero exponents $m_p$ in
  $h_p$, and so such primes will not show up in the list of good
  primes $\calP_0$, as selected in Step~\ref{setP0}.  Also, if any of the
  exponents of $h$ are equivalent to $p-1$ modulo $p$ we will not be
  able to reconstruct the exponents of $h$ from $h_p$, and we identify
  these as bad in Step~\ref{checkdeg} (by artificially marking $m_p=0$, which
  ensures they will not be added to $\calP_0$).

  Correctness of the remainder of the algorithm follows from the
  previous discussion.  

  The complexity is clearly polynomial for all steps except for
  factoring in $\QQ(\zeta_p)[y]$ (Step~\ref{factorGamma}), 
  which can be performed in
  polynomial time with the algorithm of, for example,
  \citet{Landau:1985}.
\end{proof}

As stated, the algorithm
\ref{alg:algebraic} is not actually output-sensitive, as it
requires an a priori bound $\mu$ on $\tms{h}$. To avoid this, we could
start with any small value for $\mu$, say $\tms{f}$, and after each
failure double this bound. Provided that the input polynomial $f$ is in
fact an $r$th perfect power, this process with terminate after a number
of steps polynomial in the lacunary size of the output polynomial $h$.
There are also a number of other small improvements that could be made
to increase the algorithm's efficiency, which we have omitted here
for clarity.

\subsection{Faster root computation subject to conjecture}

Algorithm~\ref{alg:algebraic} is output sensitive as the cost depends
on the sparsity of the root $h$. As discussed above, there is
considerable evidence that, roughly speaking,
the root of a sparse polynomial must always
be sparse, and so the preceding algorithm may be unconditionally
polynomial-time.

In fact, with suitable sparsity bounds we can derive a more efficient
algorithm based on Newton iteration. This approach is simpler
as it does
not rely on advanced techniques such as factoring over algebraic
extension fields. It
is also more general as it applies to fields other than $\ZZ$ and to
powers $r$ which are not prime.

Unfortunately, this algorithm is not purely output-sensitive, as it
relies on a conjecture regarding the sparsity of powers of $h$. We first
present the algorithm and prove its correctness. Then we give our modest
conjecture and use it to prove the algorithm's efficiency.

Our algorithm is essentially a Newton iteration, with special care taken to
preserve sparsity. We start with the image of $h$ modulo $x$, using the
fact that $f(0) = h(0)^r$, and at Step 
$i=1,2,\ldots,\lceil \log_2(\deg h + 1) \rceil$, we compute the
image of $h$ modulo $x^i$.

Here, and for the remainder of this section, we will assume that $f,h
\in \F[x]$ with degrees $n$ and $s$ respectively such that $f = h^r$
for $r \in \NN$ at least 2, and that the characteristic of $\F$ is
either zero or greater than $n$. As usual, we define $t = \tms{f}$.

\begin{alg}[pralg]{ComputeRootNewton}
\REQUIRE $f \in \F[x],\ r \in \NN$ such that $f$ is a perfect $r$th power
\ENSURE $h \in \F[x]$ such that $f=h^r$
\STATE $u \gets$ highest power of $x$ dividing $f$ \label{begininit}
\STATE $f_u \gets$ coefficient of $x^u$ in $f$
\STATE $g \gets f/(f_ux^u)$ \label{endinit}
\STATE $h \gets 1,\quad k \gets 1$
\WHILE{$kr \leq \deg g$}
  \STATE $\ell \gets \min\{k,(\deg g)/r + 1-k\}$ \label{setl}
  \STATE \begin{minipage}[b]{1in}%
    \[a \gets \frac{(h g - h^{r+1})\rem x^{k+\ell}}{rx^k}\]%
    \end{minipage} \label{powerh}
  \STATE $h \gets h + (a/g \bmod x^\ell)\cdot x^k$ \label{quo}
  \STATE $k \gets k + \ell$
\ENDWHILE
\STATE $b \gets $ any $r$th root of $f_u$ in $\F$ \label{rthroot}
\RETURN $bhx^{u/r}$
\end{alg}

\begin{theorem}
If $f \in \F[x]$ is a perfect $r$th power, then \ref{pralg}
returns an $h \in \F[x]$ such that $h^r = f$.
\end{theorem}

\begin{proof}
Let $u,f_u,g$ be as defined in Steps \ref{begininit}--\ref{endinit}.
Thus $f = f_u g x^u$. Now let $\hat{h}$ be some $r$th root of $f$, which we 
assume exists. If we similarly write $\hat{h} = \hat{h}_v \hat{g} x^v$, 
with $\hat{h}_v \in \F$ and $\hat{g}\in\F[x]$ such that $\hat{g}(0)=1$,
then $\hat{h}^r = \hat{h}_v^{\phantom{v}r} \hat{g}^r x^{vr}$.
Therefore $f_u$ must be a perfect $r$th power in $\F$, $r|u$, and
$g$ is a perfect $r$th power in $\F[x]$ of some polynomial with constant
coefficient equal to 1.

Denote by $h_i$ the value of $h$ at the beginning of the
$i$th iteration of the while loop. So $h_1 = 1$.  We claim that at
each iteration through Step \ref{setl}, $h_i^r \equiv g \bmod
x^k$. From the discussion above, this holds for $i=1$. 
Assuming the claim holds for all $i=1,2,\ldots,j$, we prove it
also holds for $i=j+1$.

From Step \ref{quo}, $h_{j+1} = h_j + (a/g \bmod x^l) x^k$, where
$a$ is as defined on the $j$th iteration of Step \ref{powerh}. We 
observe that
\[
h_j h_j^r \equiv h_j^{r+1} + rh_j^r (a/g \bmod x^l) x^k \mod
x^{k+\ell}.
\]
From our assumption, $h_j^r \equiv f \bmod x^k$, and $l \leq k$, so we have
\[ h_j h_{j+1}^r \equiv  h_j^{r+1} + rax^k 
\equiv h_j^{r+1} + h_jf - h_j^{r+1} 
\equiv h_j f \mod x^{k+\ell}\]
Therefore $h_{j+1}^r \equiv f \bmod x^{k+\ell}$, and so by induction the
claim holds at each step. Since the algorithm terminates when $kr > \deg g$, we
can see that the final value of $h$ is an $r$th root of $g$. Finally,
$\left(bhx^{u/r}\right)^r = f_u g x^u = f$, so the theorem holds.
\end{proof}

Algorithm \ref{pralg} will only be efficient if the low-order terms of
the polynomial power $h^{r-1}$ can be efficiently computed on
Step~\ref{powerh}. Since we know that $h$ and the low-order terms of
$h^{r-1}$ are sparse, we need only a guarantee that the
\emph{intermediate powers} will be sparse as well. This is stated in the
following modest conjecture.

\begin{conjecture} 
\label{sparsity}
For $r,s\in\NN$, if the characteristic of $\F$ is zero or greater than
$rs$, and $h\in\F[x]$ with $\deg h = s$, then
\[
\tms{h^i \bmod x^{2s}} < \tms{h^r \bmod x^{2s}}+r, 
	\qquad i=1,2,\ldots,r-1.
\]
\end{conjecture}

This corresponds to intuition and experience, as the system is still
overly constrained with only $s$ degrees of freedom. Computationally,
the conjecture has also been confirmed for all of 
the numerous examples we have tested, although a more thorough
investigation of its truth would be interesting. A weaker
inequality would suffice to prove polynomial time, but we use the
stated bounds as we believe these give more accurate complexity
measures.

The application of Conjecture~\ref{sparsity} to \ref{pralg} is given by
the following simple lemma, which 
essentially tells us that the ``error'' introduced by examining higher-order
terms of $h_1^r$ is not too dense. 

\begin{lemma} \hspace*{-6pt}\footnotemark[2]\hspace*{6pt}
\label{lemma:hr1}
Let $k,\ell \in \NN$ such that $\ell \leq k$ and $k+\ell \leq s$, and
suppose $h_1 \in \F[x]$ is the unique polynomial with degree less than $k$
satisfying $h_1^r \equiv f \bmod x^k$. Then
\[\tms{h_l^{r+1} \bmod x^{k+\ell}} \leq 2t(t+r).\]
\end{lemma}

\footnotetext[2]{Subject to the validity of Conjecture \ref{sparsity}.}

\begin{proof}
Let $h_2 \in \F[x]$ be the unique polynomial of degree less than $\ell$
satisfying $h_1 + h_2x^k \equiv h \bmod x^{k+\ell}$. Since $h^r=f$,
\[
f \equiv h_1^r + rh_1^{r-1}h_2x^k \mod x^{k+\ell}.
\]
Multiplying by $h_1$ and rearranging gives
\[
h_1^{r+1} \equiv h_1 f - r f h_2 x^k \mod x^{k+\ell}.
\]
Because $h_1 \bmod x^k$ and $h_2 \bmod x^\ell$ each have at most
$\tms{h}$ terms, which by Conjecture~\ref{sparsity} is less than $t-r$,
the total number of terms in $h_1^{r-1} \bmod x^{k+\ell}$ is less than
$2t(t-r)$.
\end{proof}

We are now ready to prove the efficiency of the algorithm, assuming the
conjecture.

\begin{theorem}\hspace*{-6pt}\footnotemark[2]\hspace*{6pt}
If $f \in \F[x]$ has degree $n$ and $t$ nonzero terms, 
then \ref{pralg} uses
$O\left( (t+r)^4 \log r \log n \right)$ operations in $\F$ and
an additional $O\left( (t+r)^4 \log r \log^2 n \right)$ 
bit operations, not counting the cost of root-finding in the base field
$\F$ on Step \ref{rthroot}.
\end{theorem}

\begin{proof}
First consider the cost of computing $h^{r+1}$ in Step \ref{powerh}. 
This will be accomplished by repeatedly squaring and multiplying by
$h$, for a total of at most $2\lfloor \log_2 (r+1) \rfloor$ multiplications.
As well, each intermediate product will have at most $\tms{f} + r <
(t+r)^2$ terms, by Conjecture \ref{sparsity}.
The number of field operations required, at each iteration, is
$O\left((t+r)^4\log r\right)$, for a total cost of
$O\left((t+r)^4\log r \log n\right)$.

Furthermore, since $k+\ell \leq 2^i$ at the $i$'th step,
for $1 \leq i < \log_2 n$, the total cost in bit operations is less than
\[\sum_{1 \leq i < \log_2 n} (t+r)^4\log_2 r i \in 
O\left( (t+r)^4 \log r \log^2 n\right).\]

In fact, this is the most costly step. The initialization
in Steps \ref{begininit}--\ref{endinit} uses only $O(t)$ operations in $\F$ and
on integers at most $n$. And the cost of computing the quotient on 
Step \ref{quo} is proportional to the cost of multiplying the quotient
and dividend, which is at most $O(t(t+r))$.
\end{proof}

When $\F = \QQ$, we must account for coefficient growth.
We use the normal notion of the size of a rational number: For
$\alpha \in \QQ$, write $\alpha = a/b$ for $a,b$ relatively prime integers. Then
define $\H(\alpha) = \max\{|a|,|b|\}$. And for $f \in \QQ[x]$ with coefficients
$c_1,\ldots,c_t \in \QQ$, write $\H(f) = \max \H(c_i)$.

Thus, the size of the lacunary representation of $f \in \QQ[x]$ is proportional
to $\tms{f},\deg f$, and $\log \H(f)$. Now we prove the bit complexity of
our algorithm is polynomial in these values, when $\F = \QQ$.

\begin{theorem}\hspace*{-6pt}\footnotemark[2]\hspace*{6pt}
Suppose $f \in \QQ[x]$ has degree $n$ and $t$ nonzero terms,
and is a perfect $r$th power. \ref{pralg}
computes an $r$th root of $f$ using
$\softO\left( t(t+r)^4 \cdot \log n \cdot \log \H(f) \right)$
bit operations.
\end{theorem}

\footnotetext[2]{Subject to the validity of Conjecture \ref{sparsity}.}

\begin{proof}
Let $h \in \QQ[x]$ such that $h^r = f$, and let $c \in \ZZ_{>0}$ be minimal
such that $ch \in \ZZ[x]$. Gau\ss's Lemma tells us that $c^r$ must
be the least positive integer such that $c^rf \in \ZZ[x]$ as well. Then, using
Theorem \ref{thm:Ztbound}, we have:
\[\H(h) \leq \norm{ch}_\infty \leq \norm{ch}_2 \leq (t\norm{c^rf}_\infty)^{1/r}
 \leq t^{1/r} \H(f)^{(t+1)/r}.\]
(The last inequality comes from the fact that the lcm of the denominators of
$f$ is at most $\H(f)^t$.)

Hence $\log \H(h) \in O\left((t\log \H(f)) / r\right)$. Clearly the
most costly step in the algorithm will still be the computation of $h_i^{r+1}$
at each iteration through Step \ref{powerh}. For simplicity in our analysis, 
we can just
treat $h_i$ (the value of $h$ at the $i$th iteration of the while loop in
our algorithm) as equal to $h$ (the \emph{actual} root of $f$), since we know
$\tms{h_i} \leq \tms{h}$ and $\H(h_i) \leq \H(h)$.

Lemma \ref{lemma:hr1} and Conjecture \ref{sparsity} tell us that
$\tau(h^i) \leq 2(t+r)^2$ for $i=1,2,\ldots,r$. To compute $h^{r+1}$, we will
actually compute $(ch)^{r+1} \in \ZZ[x]$ 
by repeatedly squaring and multiplying by $ch$,
and then divide out $c^{r+1}$. This requires at most 
$\lfloor \log_2 {r+1}\rfloor$ squares and products.

Note that $\norm{(ch)^{2i}}_\infty \leq (t+r)^2\norm{(ch)^i}_\infty^2$ and
$\norm{(ch)^{i+1}}_\infty \leq (t+r)^2 \norm{(ch)^i}_\infty \norm{ch}_\infty$.
Therefore
\[
\norm{(ch)^{i}}_\infty \leq (t+r)^{2r} \norm{ch}_\infty^r,
\quad i=1,2,\ldots,r,
\]
and thus $\log \norm{(ch)^i}_\infty \in O\left( r(t+r) + t\log \H(f) \right)$,
for each intermediate power $(ch)^i$.

Thus each of the $O\left((t+r)^4\log r\right)$ field operations at each
iteration costs at most
$O(\M(t\log \H(f) + \log r(t+r)))$
bit operations, which then gives the stated result.
\end{proof}

The method used for Step \ref{rthroot} depends on the field $\F$.  For
$\F = \QQ$, we just need to find two integer perfect roots, which can
be done in ``nearly linear'' time by the algorithm of \cite{Ber98}.
Otherwise, we can use any of the well-known fast root-finding methods
over $\F[x]$ to compute a root of $x^r - f_u$.

\subsection{Computing multivariate roots}

For the problem of computing perfect polynomial roots of multivariate polynomials, we again reduce
theproblem to a univariate one, this time employing the
well-known Kronecker substitution method.

Suppose $f,h\in\F[x_1,\ldots,x_\ell]$ and $r\in\NN$ such that $f=h^r$.
It is easily seen that
each partial degree of $f$ is exactly $r$ times the corresponding 
partial degree in $h$, that is,
$\deg_{x_i} f = r \deg_{x_i} h$, for all $r\in\{1,\ldots,\ell\}$. 

Now suppose $f$ and $r$ are given and we wish
to compute $h$. First use the relations above to compute
$d_i = \deg_{x_i} h + 1$
for each $i\in\{1,\ldots,\ell\}$. (If any $\deg_{x_i} f_i$ is not a
multiple of $r$, then $f$ must not be an $r$th power.)

Now use the Kronecker substitution and define
\[
\hat{f} = f\left(y,y^{d_1},y^{d_1d_2},\ldots,y^{d_1\cdots
d_{\ell-1}}\right)
\quad\text{and}\quad
\hat{h} = h\left(y,y^{d_1},y^{d_1d_2},\ldots,y^{d_1\cdots
d_{\ell-1}}\right),
\]
where $y$ is a new variable. Clearly $\hat{f} = \hat{h}^r$, and since
each $d_i > \deg_{x_i} h$, $h$ is easily recovered from the lacunary
representation of $\hat{h}$ in the standard way: For each non-zero term
$c\,y^e$ in $\hat{h}$, compute the digits of $e$ in the mixed radix
representation corresponding to the sequence $d_1,d_2,\ldots,d_{\ell-1}$.
That is, decompose $e$ (uniquely) as
$e = e_1 + e_2d_1 + e_3d_1d_2 + \cdots + e_\ell d_1\cdots d_{\ell-1}$
with each $e_i\in\NN$ such that $e_i<d_i$.
Then the corresponding term in $h$ is 
$c\,x_1^{e_1} \cdots x_\ell^{e_\ell}$.

Therefore we simply use either algorithm above to compute $\hat{h}$ as
the $r$th root of $\hat{f}$ over $\F[y]$, then invert the Kronecker map
to obtain $h\in\F[x_1,\ldots,x_\ell]$. The conversion steps are clearly
polynomial-time, and notice that $\log \deg \hat{f}$ is at most $\ell$
times larger than $\log \deg f$. Therefore the lacunary sizes of
$\hat{f}$ and $\hat{h}$ are polynomial in the lacunary sizes of $f$ and
$h$, and the algorithms in this section yield polynomial-time
algorithms to compute perfect $r$th roots of multivariate lacunary
polynomials.

\begin{figure*}[htb]
\begin{center}
\includegraphics[width=6cm]{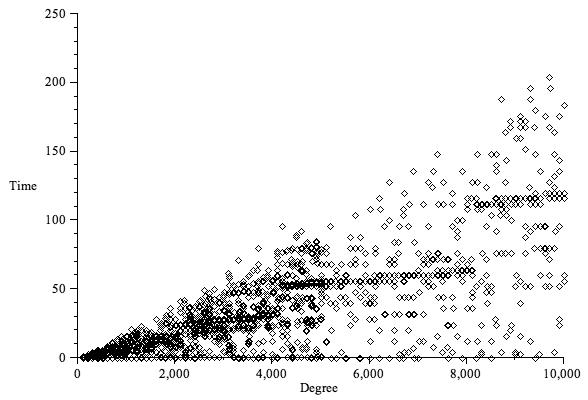}
\hspace*{1cm}
\includegraphics[width=6cm]{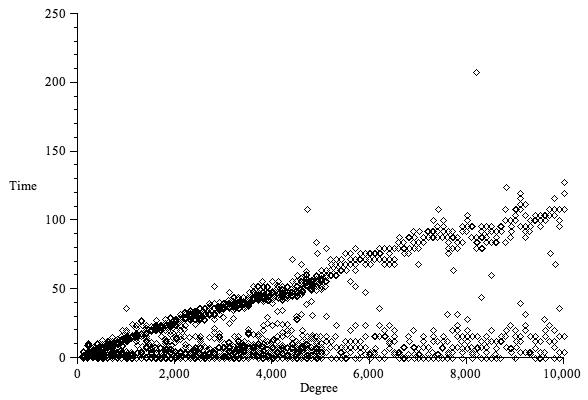}
\caption{\label{dense-tests}
Comparison of Newton Iteration (left) vs.\ our \ref{ppZ} (right).
Inputs are dense.}
\end{center}
\end{figure*}

\begin{figure*}[htb]
\begin{center}
\includegraphics[width=6cm]{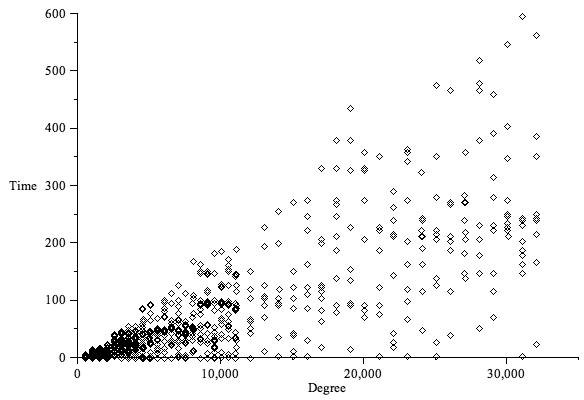}
\hspace*{1cm}
\includegraphics[width=6cm]{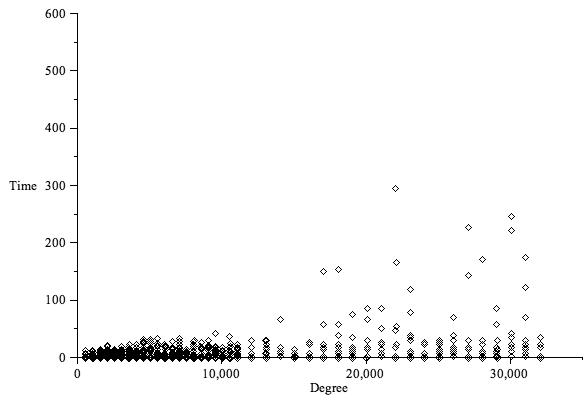}
\caption{\label{sparse-tests}
Comparison of Newton Iteration (left) vs our \ref{ppZ}
  (right).
Inputs are sparse, with sparsity fixed around 500.}
\end{center}
\end{figure*}

\section{Implementation}

To investigate the practicality of our algorithms, we implemented
\ref{ppZ} using Victor Shoup's NTL. This is a high-performance C++
for fast dense univariate polynomial computations over $\ZZ[x]$ or
$\gf_q[x]$.

NTL does not natively support a lacunary polynomial representation, so
we wrote our own using
vectors of coefficients and of exponents. In fact, since \ref{ppZ} is
a black-box algorithm, the only sparse polynomial
arithmetic we needed to implement was for evaluation at a given point.

The only significant diversion between our implementation and the
algorithm specified in Section 2 is our choice of the ground
field. Rather than working in a degree-$(r-1)$ extension of $\gf_p$,
we simply find a random $p$ in the same range such that $(r-1)\mid
p$. It is more difficult to prove that we can find such a $p$ quickly
(using e.g. the best known bounds on Linnik's Constant), but in
practice this approach is very fast because it avoids computing in field
extensions.

As a point of comparison, we also implemented the Newton iteration
approach to computing perfect polynomial roots, which appears to be the
fastest known method for dense polynomials.
This is not too dissimilar from the techniques from
the previous section on computing a lacunary $r$th root, but without
paying special attention to sparsity.  
We work modulo a randomly chosen prime $p$ to
compute an $r$th perfect root $h$, and then use random evaluations of
$h$ and the original input polynomial $f$ to certify correctness. This
yields a Monte Carlo algorithm with the same success probability as
ours, and so provides a suitable and fair comparison.

We ran two sets of tests comparing these algorithms. The first set,
depicted in Figure \ref{dense-tests}, does
not take advantage of sparsity at all; that is, the polynomials are dense and
have close to the maximal number of terms. 
It appears that the worst-case running time of our algorithm is actually a bit
better than the Newton iteration method on dense input, but on the average
they perform roughly the same. The lower triangular shape comes from the fact
that both algorithms can (and often do) terminate early. The visual gap in
the timings for the sparse algorithm comes from the fact that exactly half of
the input polynomials were perfect powers. It appears our algorithm terminates
more quickly when the polynomial is not a perfect power, but usually takes close
to the full amount of time otherwise.

The second set of tests, depicted in Figure \ref{sparse-tests},
held the number of terms of the perfect power,
$\tms{f}$, roughly fixed, 
letting the degree $n$ grow linearly. Here we can see that,
for sufficiently sparse $f$, our algorithm performs significantly 
and consistently better than the Newton iteration. 
In fact, we can see that, with some notable but rare exceptions, it appears
that the running time of our algorithm is largely independent of the degree
when the number of terms remains fixed. The outliers we see
probably come from inputs that were unluckily dense (it is not trivial to
produce examples of $h^r$ with a given fixed number of nonzero terms,
so the sparsity did vary to some extent).

Perhaps most surprisingly, although the choices of parameters for these two
algorithms only guaranteed a probability of success of at least $1/2$,
in fact over literally millions of tests performed with both algorithms
and a wide range of input polynomials, not a single failure was
recorded. This is of course due to the loose bounds employed in our
analysis, indicating a lack of understanding at some level, but it
also hints at the possibility of a deterministic algorithm, or at
least one which is probabilistic of the Las Vegas type.

Both implementations are available as C++ code downloadable from the
second author's website.

\section*{Acknowledgement}

The authors would like to thank \'Eric Schost and Pascal Koiran for
pointing out that the logarithmic derivative could be used for a
certificate of correctness in Algorithm \ref{alg:algebraic}.
The authors would also like to thank Erich Kaltofen and Igor Shparlinski
for their helpful comments.

\bibliography{lacunary-pp}

\end{document}